\documentclass[11pt]{article}
\usepackage{amsmath, mathtools,amssymb,bbm}
\makeatletter
\renewcommand*\env@matrix[1][*\c@MaxMatrixCols c]{%
  \hskip -\arraycolsep
  \let\@ifnextchar\new@ifnextchar
  \array{#1}}

\makeatother
\usepackage{mathdots}
\usepackage{mdframed}
\usepackage[margin=1in]{geometry}
\usepackage{thm-restate}
\usepackage{thmtools}
\usepackage[numbers,comma,sort&compress]{natbib}
\usepackage{subcaption}
\usepackage{authblk}
\usepackage[english]{babel}
\usepackage[utf8]{inputenc}
\usepackage{tcolorbox}
\usepackage{amsthm}
\theoremstyle{definition}
\usepackage{caption}
\usepackage{cases}
\usepackage{algorithmicx}
\usepackage{algorithm}
\usepackage[noend]{algpseudocode}
\usepackage{xcolor,soul}
\usepackage{qcircuit}
\usepackage[colorinlistoftodos, color=green!40, prependcaption]{todonotes}
\usepackage{blindtext} 
\usepackage[pdftex, pdftitle={Article}, pdfauthor={Author}, linktocpage]{hyperref} 
\hypersetup{
    colorlinks=true,
    citecolor=magenta,
    linkcolor=blue,
    filecolor=magenta,
    urlcolor=red,
}
\newtheorem{lemma}{Lemma}

\newtheorem{proposition}{Proposition}
\newtheorem{theorem}{Theorem}

\allowdisplaybreaks

\usetikzlibrary{patterns}
\usetikzlibrary{decorations.pathmorphing}
\usetikzlibrary{decorations.pathreplacing}
\usepackage{tikz-feynman}

\newcommand{\miniblockfull}{
\begin{tikzpicture}[scale=0.2, xscale=-1, baseline={([yshift=-.55ex]current bounding box.center)}]\foreach \x in {0, 1}
    {
    \foreach \y in {0, 1}
    {
    \emptysquare{\x*0.5+\y*0.25}{\y*0.5};
    }
    }\end{tikzpicture}
}

\newcommand{\miniblockthreefirst}{
\begin{tikzpicture}[scale=0.2, xscale=-1, baseline={([yshift=-.55ex]current bounding box.center)}]
\foreach \x in {0, 1}
    {
    \foreach \y in {1}
    {
    \emptysquare{\x*0.5+\y*0.25}{\y*0.5};
    }
    }
    \emptysquare{0}{0};
    \end{tikzpicture}
}

\newcommand{\miniblocktwohorizontal}{
\begin{tikzpicture}[scale=0.2, xscale=-1, baseline={([yshift=-.55ex]current bounding box.center)}]
\foreach \x in {0, 1}
    {
    \foreach \y in {0}
    {
    \emptysquare{\x*0.5+\y*0.25}{\y*0.5};
    }
    }
    \end{tikzpicture}
}

\newcommand{\miniblocktwovertical}{
\begin{tikzpicture}[scale=0.2, xscale=-1, baseline={([yshift=-.55ex]current bounding box.center)}]
\foreach \x in {0}
    {
    \foreach \y in {0,1}
    {
    \emptysquare{\x*0.5+\y*0.25}{\y*0.5};
    }
    }
    \end{tikzpicture}
}

\newcommand{\miniblockone}{
\begin{tikzpicture}[scale=0.2, xscale=-1, baseline={([yshift=-.55ex]current bounding box.center)}]
\foreach \x in {0}
    {
    \foreach \y in {0}
    {
    \emptysquare{\x*0.5+\y*0.25}{\y*0.5};
    }
    }
    \end{tikzpicture}
}

\newcommand{\snakecornerdark}[1]{
\begin{tikzpicture}[scale=0.55, baseline={([yshift=0]current bounding box.center)}]
\draw [decorate,decoration={brace,mirror,amplitude=10pt},xshift=0pt,yshift=0pt]
(0,0) -- (3,0) node [black,midway, yshift=-15pt] {\footnotesize #1};
\foreach \x in {0,...,2}
{
\foreach \y in {0,...,1}
{
\emptysquare{0.5*\x-0.25*\y}{0.5*\y};
}
}
\foreach \y in {0,...,1}
{
\emptysquare{2.5-0.25*\y}{0.5*\y};
}
\filledsquare{0}{0};
\draw[] (1.925,0.5) node[] {\footnotesize $\cdots$};
\end{tikzpicture}
}

\newcommand{\snake}[1]{
\begin{tikzpicture}[scale=0.55, baseline={([yshift=0]current bounding box.center)}]
\draw [decorate,decoration={brace,mirror,amplitude=10pt},xshift=0pt,yshift=0pt]
(0,0) -- (3,0) node [black,midway, yshift=-15pt] {\footnotesize #1};
\foreach \x in {0,...,2}
{
\foreach \y in {0,...,1}
{
\emptysquare{0.5*\x-0.25*\y}{0.5*\y};
}
}
\foreach \y in {0,...,1}
{
\emptysquare{2.5-0.25*\y}{0.5*\y};
}
\draw[] (1.925,0.5) node[] {\footnotesize $\cdots$};
\end{tikzpicture}
}

\newcommand{\snakehalf}[1]{
\begin{tikzpicture}[scale=0.55, baseline={([yshift=0]current bounding box.center)}]
\draw [decorate,decoration={brace,mirror,amplitude=10pt},xshift=0pt,yshift=0pt]
(0,0) -- (3,0) node [black,midway, yshift=-15pt] {\footnotesize #1};
\foreach \x in {0,...,2}
{
\foreach \y in {1}
{
\emptysquare{0.5*\x-0.25*\y}{0.5*\y};
}
}
\foreach \x in {0,...,2}
{
\foreach \y in {0}
{
\filledsquare{0.5*\x-0.25*\y}{0.5*\y};
}
}

\foreach \y in {1}
{
\emptysquare{2.5-0.25*\y}{0.5*\y};
}
\filledsquare{2.5}{0};
\draw[] (1.925,0.5) node[] {\footnotesize $\cdots$};
\end{tikzpicture}
}

\newcommand{\clusters}[2]{
\begin{tikzpicture}[scale=0.55, baseline={([yshift=0]current bounding box.center)}]
\draw [decorate,decoration={brace,amplitude=10pt},xshift=0pt,yshift=0pt]
(0,0) -- (-1.25,3) node [black,midway,xshift=-0.75cm, yshift=-0.15cm] {\footnotesize #2};
\draw [decorate,decoration={brace,mirror,amplitude=10pt},xshift=0pt,yshift=0pt]
(0,0) -- (3,0) node [black,midway, yshift=-15pt] {\footnotesize #1};
\foreach \x in {0,...,2}
{
\foreach \y in {0,...,2}
{
\emptysquare{0.5*\x-0.25*\y}{0.5*\y};
}
}
\foreach \x in {0,...,2}
{
\emptysquare{0.5*\x-1.25}{2.5};
}
\foreach \y in {0,...,2}
{
\emptysquare{2.5-0.25*\y}{0.5*\y};
}
\emptysquare{1.25}{2.5};
\draw[dotted, line width=0.375mm]  (-0.5, 2.25) -- (-0.125, 1.75);
\draw[dotted, line width=0.375mm]  (1.625, 2.25) -- (2, 1.75);
\draw[] (0.8,2.75) node[] {$\cdots$};
\draw[] (1.8,0.75) node[] {$\cdots$};
\end{tikzpicture}
}

\newcommand{\clustersnorth}[2]{
\begin{tikzpicture}[scale=0.55, baseline={([yshift=0]current bounding box.center)}]
\draw [decorate,decoration={brace,amplitude=10pt},xshift=0pt,yshift=0pt]
(0,0) -- (-1.25,3) node [black,midway,xshift=-0.8cm, yshift=-0.15cm] {\footnotesize #2};
\draw [decorate,decoration={brace,mirror,amplitude=10pt},xshift=0pt,yshift=0pt]
(0,0) -- (3,0) node [black,midway, yshift=-15pt] {\footnotesize #1};
\foreach \x in {0,...,2}
{
\foreach \y in {0,...,2}
{
\emptysquare{0.5*\x-0.25*\y}{0.5*\y};
}
}
\foreach \x in {0,...,2}
{
\emptysquare{0.5*\x-1.25}{2.5};
}
\foreach \y in {0,...,2}
{
\emptysquare{2.5-0.25*\y}{0.5*\y};
}
\emptysquare{1.25}{2.5};
\draw[dotted, line width=0.375mm]  (-0.5, 2.25) -- (-0.125, 1.75);
\draw[dotted, line width=0.375mm]  (1.625, 2.25) -- (2, 1.75);
\draw[] (0.8,2.75) node[] {$\cdots$};
\draw[] (1.8,0.75) node[] {$\cdots$};
\filledsquare{1.25-2.5}{2.5};
\filledsquare{1.75-2.5}{2.5};
\filledsquare{2.25-2.5}{2.5};
\filledsquare{3.75-2.5}{2.5};
\end{tikzpicture}
}

\newcommand{\clusterssouth}[2]{
\begin{tikzpicture}[scale=0.55, baseline={([yshift=0]current bounding box.center)}]
\draw [decorate,decoration={brace,amplitude=10pt},xshift=0pt,yshift=0pt]
(0,0) -- (-1.25,3) node [black,midway,xshift=-0.8cm, yshift=-0.15cm] {\footnotesize #2};
\draw [decorate,decoration={brace,mirror,amplitude=10pt},xshift=0pt,yshift=0pt]
(0,0) -- (3,0) node [black,midway, yshift=-15pt] {\footnotesize #1};
\foreach \x in {0,...,2}
{
\foreach \y in {0,...,2}
{
\emptysquare{0.5*\x-0.25*\y}{0.5*\y};
}
}
\foreach \x in {0,...,2}
{
\emptysquare{0.5*\x-1.25}{2.5};
}
\foreach \y in {0,...,2}
{
\emptysquare{2.5-0.25*\y}{0.5*\y};
}
\emptysquare{1.25}{2.5};
\draw[dotted, line width=0.375mm]  (-0.5, 2.25) -- (-0.125, 1.75);
\draw[dotted, line width=0.375mm]  (1.625, 2.25) -- (2, 1.75);
\draw[] (0.8,2.75) node[] {$\cdots$};
\draw[] (1.8,0.75) node[] {$\cdots$};
\filledsquare{0}{0};
\filledsquare{0.5}{0};
\filledsquare{1}{0};
\filledsquare{2.5}{0};
\end{tikzpicture}
}

\newcommand{\dottedsquare}[2]{\draw[pattern color=red!50!white, pattern=dots] (#1, #2) rectangle (#1 +0.5, #2 + 0.5)}
\newcommand{\emptysquare}[2]{\draw[fill=white] (#1, #2) rectangle (#1 +0.5, #2 + 0.5)}
\newcommand{\filledsquare}[2]{\draw[fill=black!25!white, draw=black] (#1, #2) rectangle (#1 +0.5, #2 + 0.5)}
\newcommand{\centeredTikZ}[1]{
\begin{tikzpicture}[scale=0.55, xscale=-1, baseline={([yshift=-.55ex]current bounding box.center)}]
#1
\end{tikzpicture}
}
\newcommand{\Prob}[1]{P\left[ #1 \right]}

\begin{document}
\title{An exactly solvable ansatz for statistical mechanics models}

\author{\normalsize Isaac H. Kim\thanks{The University of Sydney}}

\date{\today} 
\maketitle
\begin{abstract}
We propose a family of ``exactly solvable'' probability distributions to approximate partition functions of two-dimensional statistical mechanics models. While these distributions lie strictly outside the mean-field framework, their free energies can be computed in a time that scales linearly with the system size. This construction is based on a simple but nontrivial solution to the marginal problem. We formulate two non-linear constraints on the set of locally consistent marginal probabilities that simultaneously (i) ensure the existence of a consistent global probability distribution and (ii) lead to an exact expression for the maximum global entropy.
\end{abstract}                             

\vspace*{\fill}
\emph{In memory of David Poulin.}

\newpage
\tableofcontents

\vspace*{\fill}
\textit{``Three blind men were shown an elephant. They touched it with their hands to determine what the creature was. The first man felt the trunk, and claimed the elephant was like a snake. The second man touched its leg and claimed the elephant was like a tree. The third man touched its tail, and claimed that the elephant was like a slender rope.'' I nodded. ``Oh, I get it. All of them were right. All of them were wrong. They couldn't get the whole picture.'' Shirk nodded. ``Precisely. I am just another blind man. I do not get the whole picture of what transpires in all places. I am blind and limited. I would be a fool to think myself wise. And so, not knowing what the universe means, I can only try to be responsible with the knowledge, the strength, and the time given to me.”}

-- Jim Butcher, Death Masks

\newpage

\section{Introduction}
One of the fundamental tasks in physics is the calculation of partition function. Naively, this involves a summation over all possible configurations of the system, incurring a computational cost that scales exponentially with the system size. A popular alternative is to use unbiased sampling-based methods such as Markov chain Monte Carlo~\cite{Hastings1970}. While such methods do work, the sample cost can be still significant, especially if the mixing time is long.

Barring the special case of one-dimensional systems, which are amenable to the transfer matrix method~\cite{Kramers1941,Kramers1941a}, an alternative is to use approximate approaches such as mean-field theory or cluster-based methods~\cite{Bethe1935,Kikuchi1951,Yedidia2005}. While these methods are faster, they come with a price: systematic error. For example, mean-field theory becomes a valid description only in infinite-dimensional systems. In particular, making such an approximation systematically excludes systems with nontrivial spatial correlation. In the cluster-based methods~\cite{Bethe1935,Kikuchi1951}, one approximates the entropy of an ensemble by a linear combination of entropies of some marginal probabilities. For instance, the entropy of the system described in Fig.~\ref{fig:example} would be approximated by
\begin{equation*}
    H(1,\ldots, 9)\approx H(1245) + H(2356) + H(4578) + H(5689) - H(25) - H(45) - H(56) - H(58) + H(5),
\end{equation*}
in the spirit of the inclusion-exclusion principle, where $H(\cdots)$ is the entropy of the marginal probability distribution over the variables in the parenthesis.\footnote{This decomposition is motivated from the inclusion-exclusion principle applied to the four sets $\{1, 2, 4, 5\}$, $\{2, 3, 5, 6\}$, $\{4, 5, 7, 8 \}$, $\{5, 6, 8, 9\}$. } Decomposition like this is generally inexact. 
\begin{figure}[h]
    \centering
    \begin{tikzpicture}[scale=0.8]
    \foreach \x in {1,...,3}
    {
    \foreach \y in {0,...,2}
    {
    \pgfmathtruncatemacro{\result}{\x + \y*3}
    \emptysquare{\x*0.5 - 0.25*\y}{0.5*\y};
    \node[] () at (\x*0.5-0.25*\y+0.25, 0.5*\y+0.25) {\result};
    }
    }
    \end{tikzpicture}
    \caption{A system consisting of $9$ particles.}
    \label{fig:example}
\end{figure}
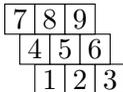

The main purpose of this paper is to introduce a new ansatz that can overcome both of these issues. We construct a family of probability distributions that can host nontrivial spatial correlations while allowing an efficient computation of energy and entropy. By minimizing the free energy within the space of such probability distributions, one can obtain a variational upper bound to the thermodynamic free energy, namely the logarithm of the partition function.

This new family has a few interesting properties. First, the probability distribution is \emph{defined} by a set of marginal probabilities. Specifically, we define our probability distribution to be the maximum-entropy probability distribution consistent with some marginals. Normally, simply specifying the marginals do not define a probability distribution because there may not exist a probability distribution consistent with the marginals. We solve this problem by introducing a special set of conditions under which one can ensure the existence of such distribution. Importantly, these conditions can be efficiently verified.

These special conditions have an intimate connection with the aforementiond inclusion-exclusion principle. Roughly speaking, the condition says that the entropy of the marginals obey the inclusion-exclusion principle \emph{internally}. Together with the condition that different marginals over the same random variables yield the same marginal probability distribution (the so called ``local consistency'' condition), the internal inclusion-exclusion principle ensures the existence of a probability distribution consistent with the given marginals. Moreover, the maximum entropy consistent with such marginals obey the inclusion-exclusion principle as well, leading to an \emph{exact} expression for the maximum entropy.

In particular, these results imply that both the energy and the maximum entropy can be decomposed into a sum of terms that can be readily computed from the given marginals. Consequently, we can compute a variational upper bound on the thermodynamic free energy in a time that scales linearly with the system size.

A construction similar to ours has appeared in a quantum setting~\cite{Kim2016}. However, that construction was extremely elaborate, limiting its practical application. Moreover, it is unclear if the entropy of the ansatz in Ref.~\cite{Kim2016} can be computed efficiently. By focusing on the classical setting, we were able to make progress on both fronts. We obtained a significantly simpler construction and moreover found an exact expression for the global entropy. Furthermore,  many of these results actually apply to the quantum setting. These developments will be discussed in our companion paper~\cite{Kim2020}.

It is also interesting to compare our work with Ref.~\cite{Wang2018}, which completely characterized translationally invariant marginals when the local variables can take a small number of possible values. While the conditions formulated in our paper does not suffer from such restrictions, we do not expect our condition to completely characterize such sets. As such, these two works are complementary to each other. Another important difference is that the conditions formulated in Ref.~\cite{Wang2018} are convex whereas our conditions are not.

The rest of this paper is structured as follows. In Section~\ref{sec:statemch}, we discuss our setup and summarize our main results. In Section~\ref{sec:marginals}, we introduce the fundamental objects of this paper, namely the marginal probabilities. In particular, we provide a concrete formulation of the translational invariance condition. In Section~\ref{sec:markovian_marginal}, we introduce an extra set of constraints that ensures the existence of a consistent global probability distribution. In Section~\ref{sec:local_extension}, we prove that the conditions imposed in Section~\ref{sec:marginals} and Section~\ref{sec:markovian_marginal} implies the existence of a global probability distribution consistent with the given marginal probabilities. In Section~\ref{sec:entropy}, we compute the exact expression for the maximum global entropy consistent with the aforementioned constraints.  We end with a discussion in Section~\ref{sec:discussion}.

\section{Summary}\label{sec:statemch}
In this section, we set up our notations and summarize our findings. We are interested in upper bounding the thermodynamic free energy density of translation-invariant statistical mechanical models in two dimensions. While our approach applies more generally, we focus on the translation-invariant case for pedagogical reasons.

To start with, we will assume that the Hamiltonian of a given model is a sum of terms each of which are supported on a ball of bounded radius. Specifically, we shall assume that each term is supported in one of the $2\times 2$ \emph{clusters}, the red shaded region in Fig.~\ref{fig:cluster0}. If not, we can always coarse-grain the system so that this condition is satisfied. 
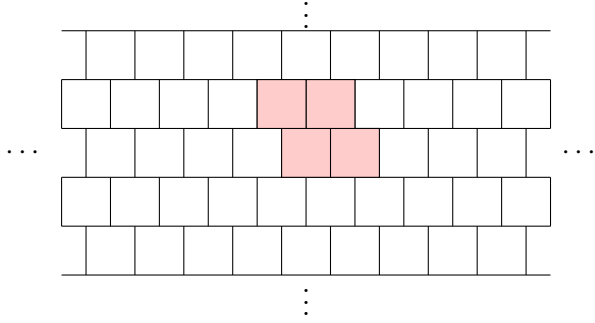
\begin{figure}[h]
    \centering
    \begin{tikzpicture}[scale=0.65, xscale=-1]
    \foreach \y in {0,...,5}
    {
    \draw[](0.25, \y + 0.25) -- (10.25, \y+0.25);
    }
    \foreach \x in {1,...,10}
    {
    \foreach \y in {1,5,9}
    {
    \draw[] (\x-0.25, \y*0.5-0.25) -- (\x-0.25, \y*0.5+0.75); 
    }
    }
    \foreach \x in {1,...,11}
    {
    \foreach \y in {3,7}
    {
    \draw[] (\x-0.75, \y*0.5-0.25) -- (\x-0.75, \y*0.5+0.75); 
    }
    }
    \draw[fill=red!20!white] (3.75, 2.25) rectangle (4.75, 3.25);
    \draw[fill=red!20!white] (4.75, 2.25) rectangle (5.75, 3.25);
    \draw[fill=red!20!white] (4.25, 3.25) rectangle (5.25, 4.25);
    \draw[fill=red!20!white] (5.25, 3.25) rectangle (6.25, 4.25);
    \node[] at (11,2.75){$\cdots$};
    \node[] at (-0.375,2.75){$\cdots$};
    \node[] at (5.25, 5.75) {$\vdots$};
    \node[] at (5.25, -0.125) {$\vdots$};
    \end{tikzpicture}
    \caption{Upon coarse-graining, every term in the Hamiltonian is contained in at least one of the $2\times 2$ clusters.}
    \label{fig:cluster0}
\end{figure}

\begin{figure}[h]
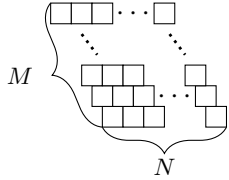

    \centering
    \clusters{$N$}{$M$}
    \caption{A $N\times M$ cluster. Our convention is to use open boundary conditions.}
    \label{fig:cluster_convention}
\end{figure}

Therefore, without loss of generality, the Hamiltonian can be written as
\begin{equation}
    H = \sum_{\miniblockfull} h_{\miniblockfull},
\end{equation}
where the summation is taken over every $2\times 2$ cluster with an appropriate boundary condition; see Fig.~\ref{fig:cluster_convention}. Assuming translation invariance of the marginal probabilities, the energy per cluster reads
\begin{equation}
    \lim_{N, M\to \infty}\frac{E}{NM} = \mathbb{E}_{P[\miniblockfull]}[h_{\miniblockfull}],
\end{equation}
where the $\mathbb{E}_{P[\miniblockfull]}[\cdots]$ is the expectation value taken over the marginal probability distribution over a cluster.

The thermodynamic free energy of a system is
\begin{equation}
    F_{N,M}(\beta) = \log Z_{N,M}(\beta),
\end{equation}
where $Z_{N,M}(\beta)=\sum_{\Omega}e^{-\beta E[\Omega]}$ is the partition function over the $N\times M$ system, where the summation is taken over a set of configurations $\{\Omega \}$. The following variational expression will be useful:
\begin{equation}
    F_{N,M}(\beta) = \min_{P[\Omega]}\left( \mathbb{E}_{P[\Omega]}[E[\Omega]] - TH(P[\Omega])\right), \label{eq:free_energy_variational}
\end{equation}
where $P[\Omega]$ is a probability distribution over $\Omega$, $\mathbb{E}_{P[\Omega]}$ is an expectation value over $P[\Omega]$, $T=1/\beta$, and $H(P[\Omega])$ is the Shannon entropy of $P[\Omega]$.

In this paper, we will be more interested in the \emph{thermodynamic free energy density}, defined as 
\begin{equation}
    f(\beta) := \lim_{N,M\to \infty} \frac{\log Z_{N,M}(\beta)}{NM}.\label{eq:variational}
\end{equation}
By restricting to a translationally invariant ansatz, we obtain
\begin{equation}
    f(\beta)= \lim_{N,M\to \infty}\left(\min_{P[\Omega] \stackrel{c}{=} P[\miniblockfull]} \left(\mathbb{E}_{P[\miniblockfull]}[h_{\miniblockfull}] - T  \frac{H(P[\Omega])}{NM} \right)\right), 
\end{equation}
where $P[\Omega] \stackrel{c}{=} P[\miniblockfull]$ means that the $P[\Omega]$ on every $2\times 2$ cluster is equal to $P[\miniblockfull]$.

There are two difficulties in using Eq.~\eqref{eq:variational}. First, ensuring $P[\Omega] \stackrel{c}{=} P[\miniblockfull]$ is known to be a hard problem.\footnote{Generally speaking, given a set of marginal probability distributions, deciding the existence of a global probability distribution consistent with those marginals is NP-hard; an efficient solution to this problem leads to an efficient algorithm for solving the $3$-coloring problem, which is NP-complete.} Second, the entropy density generally does not have a simple expression. In this paper, we will introduce a special distribution for which both of these problems can be resolved. 

\subsection{Main results}
Our main result, roughly speaking, concerns a (partial) characterization of marginal probability distributions that are compatible with a translationally invariant global state. Specifically, we formulate a marginal probability $P[\miniblockfull]$ over $2\times 2$ clusters with the conditions described below. Let $t_x$ and $t_y$ be translations in $x$- and $y$-direction by a unit spacing.\footnote{More specifically, the $y$-translation is in the direction of $(-\frac{1}{2}, \frac{\sqrt{3}}{2})$ in the convention of Fig.~\ref{fig:cluster0}} We consider the marginals that obey the following constraints:
\begin{equation}
\boxed{
    \Prob{\centeredTikZ{ 
    \emptysquare{0}{0};
    \emptysquare{0.5}{0};
    \emptysquare{0.25}{0.5};
    \emptysquare{0.75}{0.5};
    }} \text{ such that }
    \begin{cases}
    \Prob{
    \centeredTikZ
    {
        \emptysquare{0}{0};
        \filledsquare{0.5}{0};
        \emptysquare{0.25}{0.5};
        \filledsquare{0.75}{0.5};
    }
    }=
    t_x\left(
    \Prob{
    \centeredTikZ
    {
        \filledsquare{0}{0};
        \emptysquare{0.5}{0};
        \filledsquare{0.25}{0.5};
        \emptysquare{0.75}{0.5};
    }
    }\right)\\[5pt]
    \Prob{
    \centeredTikZ
    {
        \filledsquare{0}{0};
        \filledsquare{0.5}{0};
        \emptysquare{0.25}{0.5};
        \emptysquare{0.75}{0.5};
    }
    }
    =
    t_y\left(\Prob{
    \centeredTikZ
    {
        \emptysquare{0}{0};
        \emptysquare{0.5}{0};
        \filledsquare{0.25}{0.5};
        \filledsquare{0.75}{0.5};
    }
    }\right)
    \\[5pt]
    H\left(\Prob{
    \centeredTikZ
    {
        \emptysquare{0}{0};
        \emptysquare{0.25}{0.5};
    }
    } \right)
    +
    H\left(\Prob{
    \centeredTikZ
    {
        \emptysquare{0}{0};
        \emptysquare{0.5}{0};
    }
    } \right)
    -
    H\left(\Prob{
    \centeredTikZ
    {
        \emptysquare{0}{0};
    }
    } \right)
    -
    H\left(\Prob{
    \centeredTikZ
    {
        \emptysquare{0}{0};
        \emptysquare{0.5}{0};
        \emptysquare{0.75}{0.5};
    }
    } \right) = 0 \\[5pt]
     H\left(\Prob{
    \centeredTikZ
    {
        \emptysquare{0}{0};
        \emptysquare{0.25}{0.5};
    }
    } \right)
    +
    H\left(\Prob{
    \centeredTikZ
    {
        \emptysquare{0}{0};
        \emptysquare{0.5}{0};
    }
    } \right)
    -
    H\left(\Prob{
    \centeredTikZ
    {
        \emptysquare{0}{0};
    }
    } \right)
    -
    H\left(\Prob{
    \centeredTikZ
    {
        \emptysquare{0}{0};
        \emptysquare{0.25}{0.5};
        \emptysquare{0.75}{0.5};
    }
    } \right) = 0
    \end{cases}
}\label{eq:constraints_all}
\end{equation}
where the gray clusters mean that the random variables in those clusters are being summed over. It may appear that there is an ambiguity in which marginals we are referring to in the entropies. This is actually unimportant because the translation invariance condition removes such ambiguities. 

First, Eq.~\eqref{eq:constraints_all} implies the existence of a global probability distribution on a larger system which is consistent with $P[\miniblockfull]$.
\begin{restatable}[]{theorem}{theoremone}
For every $P[\miniblockfull]$ that satisfies Eq.~\eqref{eq:constraints_all}, for every integer $N, M\geq 2$, there is a probability distribution over $N\times M$ cluster (see Fig.~\ref{fig:cluster_convention}) that is consistent with $P[\miniblockfull]$ on every $2\times 2$ cluster.\label{thm:main1}
\end{restatable}
\noindent
Here, we say a probability distribution $P[\Omega]$ is consistent with $P[\miniblockfull]$ if the marginal distribution of $P[\Omega]$ over every $2\times 2$ cluster is equal to $P[\miniblockfull]$.

Our second result concerns the maximum entropy over the family of distributions discussed in Theorem~\ref{thm:main1}.
\begin{restatable}[]{theorem}{theoremtwo}
Consider a family of probability distributions over $N \times M$ clusters (see Fig.~\ref{fig:cluster_convention}) which are consistent with the $P[\miniblockfull]$ obeying Eq.~\eqref{eq:constraints_all}. The maximum entropy within this family of distributions is equal to
\begin{equation}
\begin{aligned}
   &(N-1)(M-1)H\left(\Prob{\centeredTikZ{
    \foreach \x in {0, 1}
    {
    \foreach \y in {0, 1}
    {
    \emptysquare{\x*0.5+\y*0.25}{\y*0.5};
    }
    }
    }} \right) + 
    (N-2)(M-2)H\left( 
    \Prob{\centeredTikZ{
    \foreach \x in {0}
    {
    \foreach \y in {0}
    {
    \emptysquare{\x*0.5+\y*0.25}{\y*0.5};
    }
    }
    }}
    \right) \\
    &-\left(
    (N-2)(M-1)
    H\left( 
    \Prob{\centeredTikZ{
    \foreach \x in {0}
    {
    \foreach \y in {0,1}
    {
    \emptysquare{\x*0.5+\y*0.25}{\y*0.5};
    }
    }
    }}\right)
    + (N-1)(M-2)
    H\left( 
    \Prob{\centeredTikZ{
    \foreach \x in {1,0}
    {
    \foreach \y in {0}
    {
    \emptysquare{\x*0.5+\y*0.25}{\y*0.5};
    }
    }
    }}
    \right)\right).
\end{aligned}
\end{equation}
\label{thm:main2}
\end{restatable}
\noindent
Let us emphasize that the expression in Theorem~\ref{thm:main2} is achievable by some probability distribution. This will be important for obtaining an upper bound on the thermodynamic free energy.

Let us discuss the significance of these results. First, Theorem~\ref{thm:main1} provides a positive solution to the marginal problem. Marginal problem asks, given a set of marginal probabilities, whether there is a global probability distribution that is consistent with the given marginals. While many \emph{necessary conditions} are known~\cite{Fritz2013}, less is known about sufficient conditions. Our solution is a nontrivial example, because the distribution that is consistent with $P[\miniblockfull]$ may not be a product distribution. Here is a concrete example. Let us associate a binary random variable over $\{0, 1\}$ to each cluster and let the marginal probabilities to be $\frac{1}{2}$ for the event $0000$ and $\frac{1}{2}$ for the event $1111$. Clearly, this marginal probability distribution is not a product distribution over the four random variables. Nevertheless, one can easily verify that Eq.~\eqref{eq:constraints_all} is satisfied.

Second, Theorem~\ref{thm:main1} and Theorem~\ref{thm:main2} together implies that we can obtain an interesting upper bound to the thermodynamic free energy. Instead of minimizing over all distributions in Eq.~\eqref{eq:variational}, we can minimize over a class of probability distributions obeying Eq.~\eqref{eq:constraints_all}. Moreover, for each choice of marginals, since we have to maximize the entropy, we can use the expression derived in Theorem~\ref{thm:main2}. Therefore, we obtain a \emph{locally computable} upper bound to the thermodynamic free energy density. In the thermodynamic limit, this becomes
\begin{equation}
    \lim_{N,M\to \infty} \frac{F_{N,M}(\beta)}{NM} \leq 
    \min_{\substack{P[\miniblockfull]\text{ obeying}\\ \mathcal{C}_L,\text{ } \mathcal{C}_M }}
    \widetilde{f}_{\beta}\left(\Prob{\centeredTikZ{
    \foreach \x in {0, 1}
    {
    \foreach \y in {0, 1}
    {
    \emptysquare{\x*0.5+\y*0.25}{\y*0.5};
    }
    }
    }} \right),
\end{equation}
where 
\begin{equation}
    \widetilde{f}_{\beta} \left(\Prob{\centeredTikZ{
    \foreach \x in {0, 1}
    {
    \foreach \y in {0, 1}
    {
    \emptysquare{\x*0.5+\y*0.25}{\y*0.5};
    }
    }
    }} \right)=
     \mathbb{E}_{P[\miniblockfull]}[h_{\miniblockfull}] - \frac{1}{\beta}\left(H\left(\Prob{\centeredTikZ{
    \foreach \x in {0, 1}
    {
    \foreach \y in {0, 1}
    {
    \emptysquare{\x*0.5+\y*0.25}{\y*0.5};
    }
    }
    }} \right)-
    H\left(\Prob{\centeredTikZ{
    \foreach \x in {0, 1}
    {
    \foreach \y in {1}
    {
    \emptysquare{\x*0.5+\y*0.25}{\y*0.5};
    }
    }
    \emptysquare{0.0}{0.0};
    }
    } \right) \right).
\end{equation}

Physically, Theorem~\ref{thm:main2} means that the inclusion-exclusion principle for the entropy has a ``bootstrapping nature.'' Specifically, Eq.~\eqref{eq:constraints_all} means that the entropy of a given marginal probability must be decomposable into a linear combination of entropies of its own marginals, obeying the inclusion-exclusion principle.  Theorem~\ref{thm:main2} says that this principle applies to the maximum entropy consistent with these marginals on a much larger system. Remarkably, to arrive at this conclusion, all we needed was Eq.~\eqref{eq:constraints_all}, without invoking any global assumption.  

\section{Marginals} \label{sec:marginals}
In this section, we discuss the fundamental objects behind our construction. These are the marginal probabilities over $2\times 2$ clusters:
\begin{equation}
    \mathcal{M} := \left\{\Prob{
    \centeredTikZ
    {
        \emptysquare{0}{0};
        \emptysquare{0.5}{0};
        \emptysquare{0.25}{0.5};
        \emptysquare{0.75}{0.5};
    }
    } \right\}.
\end{equation}
Because of translational invariance, this set contains only one element.

Moreover, translational invariance imposes an extra set of linear constraints. We have summarized these constraints below, together with the ``obvious'' constraints, namely the fact that the probabilities must be nonnegative and sum to $1$.
\begin{numcases}{\mathcal{C}_L:}
    \Prob{
    \centeredTikZ
    {
        \emptysquare{0}{0};
        \emptysquare{0.5}{0};
        \emptysquare{0.25}{0.5};
        \emptysquare{0.75}{0.5};
    }
    } \geq 0 \label{constraint:positivity}\\
    \Prob{
    \centeredTikZ
    {
        \filledsquare{0}{0};
        \filledsquare{0.5}{0};
        \filledsquare{0.25}{0.5};
        \filledsquare{0.75}{0.5};
    }
    }=1\label{constraint:normalization}\\
    \Prob{
    \centeredTikZ
    {
        \emptysquare{0}{0};
        \emptysquare{0.5}{0};
        \filledsquare{0.25}{0.5};
        \filledsquare{0.75}{0.5};
    }
    }=
    \Prob{
    \centeredTikZ
    {
        \filledsquare{0}{0};
        \filledsquare{0.5}{0};
        \emptysquare{0.25}{0.5};
        \emptysquare{0.75}{0.5};
    }
    }\label{constraint:horizontal} \\
    \Prob{
    \centeredTikZ
    {
        \emptysquare{0}{0};
        \filledsquare{0.5}{0};
        \emptysquare{0.25}{0.5};
        \filledsquare{0.75}{0.5};
    }
    }=
    \Prob{
    \centeredTikZ
    {
        \filledsquare{0}{0};
        \emptysquare{0.5}{0};
        \filledsquare{0.25}{0.5};
        \emptysquare{0.75}{0.5};
    }
    }\label{constraint:vertical}
\end{numcases}
As before, the shaded cluster means that we are summing over all the configurations in that cluster. 

The first condition is the nonnegativitiy of the probability amplitudes and the second condition is the normalization of the probability distribution. The third and the fourth condition represent the translational invariance constraint, in the vertical and horizontal direction respectively.

As a side note, we emphasize that it is quite simple to relax the translational invariance condition. To do so, it suffices to replace the cluster appearing on the right-hand-side of Eq.~\eqref{constraint:horizontal} to a $2\times 2$ cluster shifted in the $y$-direction by $-1$. Also, the cluster appearing on the right-hand-side of Eq.~\eqref{constraint:vertical} can be replaced by a $2\times 2$ cluster shifted in the $x$-direction by $-1$.

An important consequence of Eq.~\eqref{constraint:horizontal} and Eq.~\eqref{constraint:vertical} is that the marginal probabilities over every cluster is identical. To see why, let us denote the variables from left to right, from bottom to top, as $X, Y, Z,$ and $W$. Eq.~\eqref{constraint:horizontal} implies
\begin{equation}
\begin{aligned}
    P[X] &= P[Z], \\
    P[Y] &= P[W].
\end{aligned}
\end{equation}
and Eq.~\eqref{constraint:vertical} implies
\begin{equation}
\begin{aligned}
    P[X] &= P[Y], \\
    P[Z] &= P[W].
\end{aligned}
\end{equation}
Therefore, 
\begin{equation}
    P[X] = P[Y] = P[Z] = P[W].
\end{equation}

\section{Markovian marginals}\label{sec:markovian_marginal}
Now, we will impose extra \emph{non-linear} constraints on the marginals. Such constraints were motivated from the study of quantum many-body systems~\cite{Kim2016}, but our analysis here is significantly simplified. 

We will impose the following two constraints on the marginals:
\begin{numcases}{\mathcal{C}_M:}
    H\left(\Prob{
    \centeredTikZ
    {
        \emptysquare{0}{0};
        \emptysquare{0.25}{0.5};
    }
    } \right)
    +
    H\left(\Prob{
    \centeredTikZ
    {
        \emptysquare{0}{0};
        \emptysquare{0.5}{0};
    }
    } \right)
    -
    H\left(\Prob{
    \centeredTikZ
    {
        \emptysquare{0}{0};
    }
    } \right)
    -
    H\left(\Prob{
    \centeredTikZ
    {
        \emptysquare{0}{0};
        \emptysquare{0.5}{0};
        \emptysquare{0.75}{0.5};
    }
    } \right) = 0, \label{constraint:markov1} \\
     H\left(\Prob{
    \centeredTikZ
    {
        \emptysquare{0}{0};
        \emptysquare{0.25}{0.5};
    }
    } \right)
    +
    H\left(\Prob{
    \centeredTikZ
    {
        \emptysquare{0}{0};
        \emptysquare{0.5}{0};
    }
    } \right)
    -
    H\left(\Prob{
    \centeredTikZ
    {
        \emptysquare{0}{0};
    }
    } \right)
    -
    H\left(\Prob{
    \centeredTikZ
    {
        \emptysquare{0}{0};
        \emptysquare{0.25}{0.5};
        \emptysquare{0.75}{0.5};
    }
    } \right) = 0, \label{constraint:markov2}
\end{numcases}
where the specified marginal probabilities can be obtained straightforwardly from the marginal over the $2\times 2$ cluster by summing over the appropriate set of variables.

In fact, we can reduce this to a single constraint. Because both linear combinations of entropies must be nonnegative,\footnote{For any probability distribution over three variables $X,Y,$ and $Z$, one can show that $S(P[XY]) + S(P[YZ]) - S(P[Y]) - S(P[XYZ])\geq 0$.} the following condition implies both Eq.~\eqref{constraint:markov1} and Eq.~\eqref{constraint:markov2}.
\begin{equation}
2H\left(\Prob{
    \centeredTikZ
    {
        \emptysquare{0}{0};
        \emptysquare{0.25}{0.5};
    }
    } \right)
    +
    2H\left(\Prob{
    \centeredTikZ
    {
        \emptysquare{0}{0};
        \emptysquare{0.5}{0};
    }
    } \right)
    -
    2H\left(\Prob{
    \centeredTikZ
    {
        \emptysquare{0}{0};
    }
    } \right)
    -
    H\left(\Prob{
    \centeredTikZ
    {
        \emptysquare{0}{0};
        \emptysquare{0.25}{0.5};
        \emptysquare{0.75}{0.5};
    }
    } \right)
    -
    H\left(\Prob{
    \centeredTikZ
    {
        \emptysquare{0}{0};
        \emptysquare{0.5}{0};
        \emptysquare{0.75}{0.5};
    }
    } \right)=0
    \label{constraint:markov_combined}
\end{equation}

An important consequence of these constraints is that the probability distribution over these variables have an internal Markov chain structure. Specifically, 
\begin{numcases}{\mathcal{C}_M:}
    \begin{aligned} \Prob
    {
    \centeredTikZ
    {
    \emptysquare{0}{0};
    \emptysquare{0.5}{0};
    \emptysquare{0.75}{0.5};
    }
    }
    =
    \Prob
    {
    \centeredTikZ
    {
    \emptysquare{0}{0};
    \emptysquare{0.5}{0};
    \dottedsquare{0.75}{0.5};
    }
    }
    \Prob
    {
    \centeredTikZ
    {
    \dottedsquare{0}{0};
    \emptysquare{0.5}{0};
    \emptysquare{0.75}{0.5};
    }
    }
    \Prob
    {
    \centeredTikZ
    {
    \dottedsquare{0}{0};
    \emptysquare{0.5}{0};
    \dottedsquare{0.75}{0.5};
    }
    }^{-1}, \label{constraint:markov_alternative1} \\
    \Prob
    {
    \centeredTikZ
    {
    \emptysquare{0}{0};
    \emptysquare{0.5}{0};
    \emptysquare{-0.25}{-0.5};
    }
    }
    =
    \Prob
    {
    \centeredTikZ
    {
    \emptysquare{0}{0};
    \dottedsquare{0.5}{0};
    \emptysquare{-0.25}{-0.5};
    }
    }
    \Prob
    {
    \centeredTikZ
    {
    \emptysquare{0}{0};
    \emptysquare{0.5}{0};
    \dottedsquare{-0.25}{-0.5};
    }
    }
    \Prob
    {
    \centeredTikZ
    {
    \emptysquare{0}{0};
    \dottedsquare{0.5}{0};
    \dottedsquare{-0.25}{-0.5};
    }
    }^{-1},\label{constraint:markov_alternative2}
    \end{aligned}
\end{numcases}
where again $P[\cdots]$ is a probability distribution over the set of white clusters in the square bracket. The dotted squares are bookkeeping devices, to specify the relative location of these clusters. 

\subsection{What do the Markovian constraints mean?}
We shall soon see that these seemingly mysterious constraints have remarkable consequences. But for now, we should ask: when do these constraints become reasonable?

To answer this question, it is helpful to imagine the thermodynamic limit, the limit in which individual cluster contains many degrees of freedom. In that setup, we can hypothesize that the entropy of the marginal distributions will obey a scaling law. Specifically, one may assume that the entropy can be decomposed as follows:
\begin{equation}
    H\left(\Prob
    {
    \centeredTikZ
    {
    \emptysquare{0}{0};
    \emptysquare{0.25}{0.5};
    \emptysquare{0.75}{0.5};
    }
    } \right) = \int_{\miniblockthreefirst} h_{\miniblockthreefirst, V} dV + \int_{\partial \left( \miniblockthreefirst \right)} h_{\miniblockthreefirst, A} dA + \ldots,
\end{equation}
where $h_{\miniblockthreefirst, V}$ and $h_{\miniblockthreefirst, A}$ represent the entropy density in the bulk and the boundary of $\miniblockthreefirst$ respectively and the ellipsis represents a subleading contribution. The first integral represents the ``bulk'' contribution and the second integral represents the ``boundary'' contribution to the entropy.

Similarly, we can express the other terms in the constraints as follows.
\begin{equation}
\begin{aligned}
    H\left(\Prob
    {
    \centeredTikZ
    {
    \emptysquare{0}{0};
    \emptysquare{0.5}{0};
    }
    } \right) &= \int_{\miniblocktwohorizontal} h_{\miniblocktwohorizontal, V} dV + \int_{\partial \left( \miniblocktwohorizontal \right)} h_{\miniblocktwohorizontal, A} dA + \ldots \\
    H\left(\Prob
    {
    \centeredTikZ
    {
    \emptysquare{0}{0};
    \emptysquare{0.25}{0.5};
    }
    } \right) &= \int_{\miniblocktwovertical} h_{\miniblocktwovertical, V} dV + \int_{\partial \left( \miniblocktwovertical \right)} h_{\miniblocktwovertical, A} dA + \ldots \\
    H\left(\Prob
    {
    \centeredTikZ
    {
    \emptysquare{0}{0};
    }
    } \right) &= \int_{\miniblockone} h_{\miniblockone, V} dV + \int_{\partial \left( \miniblockone \right)} h_{\miniblockone, A} dA + \ldots
\end{aligned}
\end{equation}

If the densities do not depend greatly on the underlying clusters, which we may expect to be the case when the clusters are sufficiently large, then these leading contributions cancel out each other. Similarly, the subleading area terms also cancel out. The remaining $\mathcal{O}(1)$ term, if they do not depend on the fine details of the geometry, should also cancel out. Therefore, under this somewhat speculative scaling law, we can conclude that Eq.~\eqref{constraint:markov1} and Eq.~\eqref{constraint:markov2} become valid approximations. 

This is the main rationale on why we impose these constraints. A remarkable thing is that these constraints actually have nontrivial implications. While $\mathcal{C}_L$ by itself does not imply that there is some global probability distribution $P[\Omega]$ consistent with the marginals, together with $\mathcal{C}_M$, we can ensure the existence of such probability distribution. Moreover, the maximum entropy consistent with the marginals obeying $\mathcal{C}_L$ and $\mathcal{C}_M$ has a closed-form expression. These results will be discussed in Section~\ref{sec:local_extension} and Section~\ref{sec:entropy}.

\section{Global consistency}
\label{sec:local_extension}

So far, we have introduced two types of constraints, which we referred to as $\mathcal{C}_L$ and $\mathcal{C}_M$. In this section, we show that these constraints imply the existence of a probability distribution over every $N\times M$ cluster (Fig.~\ref{fig:cluster_convention}) for every integer $N, M \geq 2$ such that its marginal probabilities over every $2\times 2$ cluster is identical to the given marginals. To refer to each cluster, we will use the convention described in Fig.~\ref{fig:dependency}.
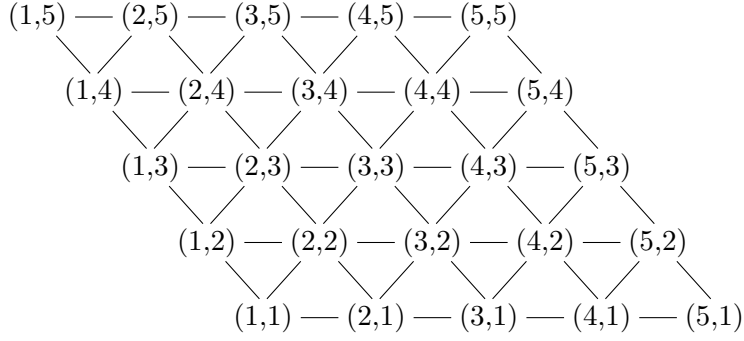
\begin{figure}[h]
    \centering
    \begin{tikzpicture}
    \foreach \x in {1, ..., 5}
    {
    \foreach \y in {1, ..., 5}
    {
    \node[] () at (1.5*\x-0.75*\y, \y) {(\x,\y)};
    }
    }
    \foreach \x in {1,...,4}
    {
    \foreach \y in {1,...,5}
    {
    \draw[-] (1.5*\x-0.75*\y + 0.5, \y) -- (1.5*\x - 0.75*\y + 1, \y);
    }
    }
    
    \foreach \x in {1,...,5}
    {
    \foreach \y in {1,...,4}
    {
    \draw[-] (1.5*\x-0.75*\y-0.05, \y+0.25) -- (1.5*\x-0.75*\y -0.5, \y+0.75);
    }
    }
    
    \foreach \x in {1,...,4}
    {
    \foreach \y in {1,...,4}
    {
    \draw[-] (1.5*\x-0.75*\y+0.05, \y+0.25) -- (1.5*\x-0.75*\y + 0.5, \y+0.75);
    }
    }

    \end{tikzpicture}
    \caption{Our convention for the clusters, for a $5\times 5$ system.}
    \label{fig:dependency}
\end{figure}

We shall label the clusters within this collection by its corners. For example, 
\begin{equation}
\Prob{
\centeredTikZ{
\foreach \x in {0,1}
    {
    \foreach \y in {0,1}
    {
    \emptysquare{0.5*\x+0.25*\y}{0.5*\y};
    }
    }
    \draw[->] (0.25, -0.5) -- (0.25, 0);
    \draw[] (0.5, -0.5)  -- (0.25,-0.5) node[pos=0, left]{$\vec{v}$};
}
}
\end{equation}
is a probability distribution of a $2\times 2$ cluster whose corner is located at $\vec{v} = (v_x, v_y)$. We shall consider the following object:
\begin{equation}
\boxed{
    \Prob{\clusters{$N$}{$M$}} := 
\frac{\displaystyle \left(\prod_{\substack{v_x \in \{2,\ldots, N \} \\ v_y \in \{1,\ldots, M-1 \}}}
\Prob{
\centeredTikZ{
\foreach \x in {0,1}
    {
    \foreach \y in {0,1}
    {
    \emptysquare{0.5*\x+0.25*\y}{0.5*\y};
    }
    }
    \draw[->] (0.25, -0.5) -- (0.25, 0);
    \draw[] (0.5, -0.5)  -- (0.25,-0.5) node[pos=0, left]{$\vec{v}$};
}
}
\right)
\left(\prod_{\substack{v_x \in \{2,\ldots, N-1 \} \\ v_y \in \{2,\ldots, M-1 \}}}
\Prob{
\centeredTikZ{
\foreach \x in {0}
    {
    \foreach \y in {0}
    {
    \emptysquare{0.5*\x+0.25*\y}{0.5*\y};
    }
    }
    \draw[->] (0.25, -0.5) -- (0.25, 0);
    \draw[] (0.5, -0.5)  -- (0.25,-0.5) node[pos=0, left]{$\vec{v}$};
}
}
\right)
}
{
\left( \displaystyle \prod_{\substack{v_x \in \{2, \ldots, N-1 \} \\ v_y \in \{1,\ldots, M-1 \}}}
\Prob{
\centeredTikZ{
\foreach \x in {0}
    {
    \foreach \y in {0,1}
    {
    \emptysquare{0.5*\x+0.25*\y}{0.5*\y};
    }
    }
    \draw[->] (0.25, -0.5) -- (0.25, 0);
    \draw[] (0.5, -0.5)  -- (0.25,-0.5) node[pos=0, left]{$\vec{v}$};
}
}\right)
\left(
\displaystyle \prod_{\substack{v_x \in \{2,\ldots, N \} \\ v_y \in \{2, \ldots, M-1 \}}}
\Prob{
\centeredTikZ{
\foreach \x in {0,1}
    {
    \foreach \y in {0}
    {
    \emptysquare{0.5*\x+0.25*\y}{0.5*\y};
    }
    }
    \draw[->] (0.25, -0.5) -- (0.25, 0);
    \draw[] (0.5, -0.5)  -- (0.25,-0.5) node[pos=0, left]{$\vec{v}$};
}
}
\right)
}.} \label{eq:consistent_global_state}
\end{equation}
\noindent
for $N, M \geq 3$. When $M=2$, we define it to be an object called \emph{snake}, defined below:
\begin{equation}
\boxed{
    \Prob{\snake{$N$}}:=
    \frac{\displaystyle
    \prod_{\substack{v_x \in \{2,\ldots, N \},\\ v_y=1}}
    \Prob{
\centeredTikZ{
\foreach \x in {0,1}
    {
    \foreach \y in {0,1}
    {
    \emptysquare{0.5*\x+0.25*\y}{0.5*\y};
    }
    }
    \draw[->] (0.25, -0.5) -- (0.25, 0);
    \draw[] (0.5, -0.5)  -- (0.25,-0.5) node[pos=0, left]{$\vec{v}$};
}
}}
{\displaystyle
\prod_{\substack{v_x \in \{2,\ldots, N-1 \},\\v_y=1}}
    \Prob{
\centeredTikZ{
\foreach \x in {0}
    {
    \foreach \y in {0,1}
    {
    \emptysquare{0.5*\x+0.25*\y}{0.5*\y};
    }
    }
    \draw[->] (0.25, -0.5) -- (0.25, 0);
    \draw[] (0.5, -0.5)  -- (0.25,-0.5) node[pos=0, left]{$\vec{v}$};
}
}
}.
}
\label{eq:horizontal_snake}
\end{equation}

It should be clear that Eq.~\eqref{eq:horizontal_snake} is consistent with $P[\miniblockfull]$. Our goal is to show that the object defined in Eq.~\eqref{eq:consistent_global_state} is a probability distribution consistent with $P[\miniblockfull]$, by showing that it reduces to a snake upon summing over an appropriate set of variables.

For that purpose, we will mainly focus on proving the following theorem.
\begin{theorem}
(Row reduction)
\begin{equation}
\begin{aligned}
    \Prob{\clustersnorth{$N$}{$M$}} &= \Prob{\clusters{$N$}{$M-1$}} \\
    \Prob{\clusterssouth{$N$}{$M$}} &= \Prob{\clusters{$N$}{$M-1$}}.
\end{aligned}
\end{equation}
\label{thm:reduction}
\end{theorem}

The proof mainly follows from the fact that the snake remains to be a Markov chain after summing over the variables on a single row.
\begin{lemma}
\label{lemma:snake_lemma}
\begin{equation}
\boxed{
    \Prob{\snakehalf{N}}=\frac{\displaystyle
    \prod_{\substack{v_x \in \{2,\ldots, N \},\\ v_y=2}}
    \Prob{
\centeredTikZ{
\foreach \x in {0,1}
    {
    \foreach \y in {0}
    {
    \emptysquare{0.5*\x+0.25*\y}{0.5*\y};
    }
    }
    \draw[->] (0.25, -0.5) -- (0.25, 0);
    \draw[] (0.5, -0.5)  -- (0.25,-0.5) node[pos=0, left]{$\vec{v}$};
}
}}
{\displaystyle
\prod_{\substack{v_x \in \{2,\ldots, N-1 \},\\v_y=2}}
    \Prob{
\centeredTikZ{
\foreach \x in {0}
    {
    \foreach \y in {0}
    {
    \emptysquare{0.5*\x+0.25*\y}{0.5*\y};
    }
    }
    \draw[->] (0.25, -0.5) -- (0.25, 0);
    \draw[] (0.5, -0.5)  -- (0.25,-0.5) node[pos=0, left]{$\vec{v}$};
}
}
}}.\label{eq:snake_markov}
\end{equation}
\end{lemma}
\begin{proof}
The proof is based on induction. We will set out to prove that
\begin{equation}
\Prob{\snakecornerdark{$N$}} = \Prob{\snake{$N-1$}}
\Prob{
\centeredTikZ{
\foreach \x in {0,1}
    {
    \foreach \y in {0}
    {
    \emptysquare{0.5*\x+0.25*\y}{0.5*\y};
    }
    }
    \draw[->] (0.25, -0.5) -- (0.25, 0);
    \draw[] (0.5, -0.5)  -- (0.25,-0.5) node[pos=0, left]{\tiny $(1,2)$};
}
}
\Prob{
\centeredTikZ{
\foreach \x in {0}
    {
    \foreach \y in {0}
    {
    \emptysquare{0.5*\x+0.25*\y}{0.5*\y};
    }
    }
    \draw[->] (0.25, -0.5) -- (0.25, 0);
    \draw[] (0.5, -0.5)  -- (0.25,-0.5) node[pos=0, left]{\tiny $(2,2)$};
}
}^{-1}
, \label{eq:induction}
\end{equation}
where the snake on the right hand side is defined over the set of cluster $\{v : v_x \in \{1,\ldots, N-1\}, v_y\in \{1,2\}\}$. To prove Eq.~\eqref{eq:induction}, note that
\begin{equation}
\begin{aligned}
    \Prob{
    \centeredTikZ{
    \emptysquare{0}{0};
    \filledsquare{0.5}{0};
    \emptysquare{0.25}{0.5};
    \emptysquare{0.75}{0.5};
    }
    }
    &= 
    \Prob{
    \centeredTikZ{
    \emptysquare{0}{0};
    \dottedsquare{0.5}{0};
    \emptysquare{0.25}{0.5};
    \emptysquare{0.75}{0.5};
    }
    } \\
  &=\Prob{
    \centeredTikZ{
    \dottedsquare{0}{0};
    \dottedsquare{0.5}{0};
    \emptysquare{0.25}{0.5};
    \emptysquare{0.75}{0.5};
    }
    }
    \Prob{
    \centeredTikZ{
    \emptysquare{0}{0};
    \dottedsquare{0.5}{0};
    \emptysquare{0.25}{0.5};
    \dottedsquare{0.75}{0.5};
    }
    }
    \Prob{
    \centeredTikZ{
    \dottedsquare{0}{0};
    \dottedsquare{0.5}{0};
    \emptysquare{0.25}{0.5};
    \dottedsquare{0.75}{0.5};
    }
    }^{-1}.
\end{aligned}
\end{equation}
Plugging in this identity to the definition of the snake, we obtain Eq.~\eqref{eq:induction}, from which the main claim readily follows.
\end{proof}

Next, we note the following decomposition. The proof follows from the definition of the snake and Lemma~\ref{lemma:snake_lemma}.
\begin{proposition}
\begin{equation}
\boxed{
\Prob{\clusters{$N$}{$M$}} = \Prob{\clusters{$N$}{$M-1$}} \frac{\Prob{\snake{$N$}}}{\Prob{\snakehalf{$N$}}},
}
\end{equation}
where the probability distribution over the three terms on the right hand side are over
\begin{itemize}
    \item $\{v : v_x \in \{1,N\}, v_y \in \{2, M\} \}$ for the first term,
    \item $\{v : v_x \in \{1,N\}, v_y \in \{1, 2\} \}$ for the term in the numerator,
    \item and $\{v : v_x \in \{1,N\}, v_y \in \{2\} \}$ for the tern in the denominator.
\end{itemize}
\label{prop:row_reduction}
\end{proposition}
\noindent
By summing over the variables in the clusters on the bottom row, we can prove the last identity in Theorem~\ref{thm:reduction}. Up to a global rotation of $\pi$, the same proof applies to the second identity, too.

Thus, we can complete the proof of Theorem~\ref{thm:main1} as follows. Without loss of generality, consider a $2\times 2$ cluster. By repeatedly applying Theorem~\ref{thm:reduction} to all the rows that do not overlap with this cluster, we obtain a snake; see Eq.~\eqref{eq:horizontal_snake}. From this explicit form, one can immediately verify that the marginal on the $2\times 2$ cluster is equal to $\Prob{\miniblockfull}$. Thus, we conclude with the following theorem.

\theoremone*

\section{Entropy}
\label{sec:entropy}

Because the object defined in Eq.~\eqref{eq:consistent_global_state} is a probability distribution, we can define its entropy. Moreover, the entropy has a \emph{local} decomposition. Let $H(N,M)$ be the Shannon entropy of the probability distribution in Eq.~\eqref{eq:consistent_global_state}. We obtain:
\begin{equation}
\begin{aligned}
    H(N,M) &= (N-1)(M-1)H\left(\Prob{\centeredTikZ{
    \foreach \x in {0, 1}
    {
    \foreach \y in {0, 1}
    {
    \emptysquare{\x*0.5+\y*0.25}{\y*0.5};
    }
    }
    }} \right) + 
    (N-2)(M-2)H\left( 
    \Prob{\centeredTikZ{
    \foreach \x in {0}
    {
    \foreach \y in {0}
    {
    \emptysquare{\x*0.5+\y*0.25}{\y*0.5};
    }
    }
    }}
    \right) \\
    &-\left(
    (N-2)(M-1)
    H\left( 
    \Prob{\centeredTikZ{
    \foreach \x in {0}
    {
    \foreach \y in {0,1}
    {
    \emptysquare{\x*0.5+\y*0.25}{\y*0.5};
    }
    }
    }}\right)
    + (N-1)(M-2)
    H\left( 
    \Prob{\centeredTikZ{
    \foreach \x in {1,0}
    {
    \foreach \y in {0}
    {
    \emptysquare{\x*0.5+\y*0.25}{\y*0.5};
    }
    }
    }}
    \right)\right).\label{eq:entropy_expression}
\end{aligned}
\end{equation}
Moreover,
\begin{equation}
\begin{aligned}
    \lim_{N,M\to \infty} \frac{H(N,M)}{NM}&= H\left(\Prob{\centeredTikZ{
    \foreach \x in {0, 1}
    {
    \foreach \y in {0, 1}
    {
    \emptysquare{\x*0.5+\y*0.25}{\y*0.5};
    }
    }
    }} \right) + 
    H\left( 
    \Prob{\centeredTikZ{
    \foreach \x in {0}
    {
    \foreach \y in {0}
    {
    \emptysquare{\x*0.5+\y*0.25}{\y*0.5};
    }
    }
    }}\right) -
    H\left( 
    \Prob{\centeredTikZ{
    \foreach \x in {0}
    {
    \foreach \y in {0,1}
    {
    \emptysquare{\x*0.5+\y*0.25}{\y*0.5};
    }
    }
    }}\right) -H\left( 
    \Prob{\centeredTikZ{
    \foreach \x in {1,0}
    {
    \foreach \y in {0}
    {
    \emptysquare{\x*0.5+\y*0.25}{\y*0.5};
    }
    }
    }}
    \right) \\
    &=H\left(\Prob{\centeredTikZ{
    \foreach \x in {0, 1}
    {
    \foreach \y in {0, 1}
    {
    \emptysquare{\x*0.5+\y*0.25}{\y*0.5};
    }
    }
    }} \right)-
    H\left(\Prob{\centeredTikZ{
    \foreach \x in {0, 1}
    {
    \foreach \y in {1}
    {
    \emptysquare{\x*0.5+\y*0.25}{\y*0.5};
    }
    }
    \emptysquare{1}{1};
    }
    } \right).
    \end{aligned}
\end{equation}

One may worry that Eq.~\eqref{eq:entropy_expression} is merely an entropy of \emph{some} state obeying $\mathcal{C}_L$ and $\mathcal{C}_M$. However, one can prove that this is the maximum entropy consistent with the given marginals by using the following inequality:
\begin{equation}
    H(P[XY]) + H(P[YZ]) - H(P[Y])- H(P[XYZ]) \geq 0
    \label{eq:ssa}
\end{equation}
for any probability distribution over three random variables $X,Y,$ and $Z$. 

Specifically, first consider a probability distribution over the first two rows, which is consistent with the given marginals. Then, by Eq.~\eqref{eq:ssa}, the entropy of this probability distribution is upper bounded by
\begin{equation}
H(P[2]) \leq  (N-1)H\left(\Prob{\centeredTikZ{
    \emptysquare{0}{0};
    \emptysquare{0.5}{0};
    \emptysquare{0.25}{0.5};
    \emptysquare{0.75}{0.5};
    }}\right) - (N-2) H\left(\Prob{\centeredTikZ{
    \emptysquare{0}{0};
    \emptysquare{0.25}{0.5};
    }}\right),
    \label{eq:med_initial}
\end{equation}
where $P[2]$ can be \emph{any} probability distribution over the first two rows. Moreover, we have the following recursive inequality:
\begin{equation}
   H(P[k+1]) \leq H(P[k]) + 
   (N-1)H\left(\Prob{\centeredTikZ{
    \emptysquare{0}{0};
    \emptysquare{0.5}{0};
    \emptysquare{0.25}{0.5};
    \emptysquare{0.75}{0.5};
    }}\right) - (N-2) H\left(\Prob{\centeredTikZ{
    \emptysquare{0}{0};
    \emptysquare{-0.5}{0};
    \emptysquare{0.25}{0.5};
    }}\right) - H(\Prob{\centeredTikZ{
    \emptysquare{0}{0};
    \emptysquare{0.5}{0};
    }}) 
    \label{eq:med}
\end{equation}
where $P[k]$ can be any probability distribution over the first $k$ rows. From Eq.~\eqref{eq:med_initial} and Eq.~\eqref{eq:med}, we obtain the following inequality for a general probability distribution $P[\Omega]_{N\times M}$ over $N\times M$ clusters which is consistent with $P[\miniblockfull]$:
\begin{equation}
\begin{aligned}
    H(P[\Omega]_{N\times M}) &\leq (N-1)(M-1)H\left(\Prob{\centeredTikZ{
    \foreach \x in {0, 1}
    {
    \foreach \y in {0, 1}
    {
    \emptysquare{\x*0.5+\y*0.25}{\y*0.5};
    }
    }
    }} \right) + 
    (N-2)(M-2)H\left( 
    \Prob{\centeredTikZ{
    \foreach \x in {0}
    {
    \foreach \y in {0}
    {
    \emptysquare{\x*0.5+\y*0.25}{\y*0.5};
    }
    }
    }}
    \right) \\
    &-\left(
    (N-2)(M-1)
    H\left( 
    \Prob{\centeredTikZ{
    \foreach \x in {0}
    {
    \foreach \y in {0,1}
    {
    \emptysquare{\x*0.5+\y*0.25}{\y*0.5};
    }
    }
    }}\right)
    + (N-1)(M-2)
    H\left( 
    \Prob{\centeredTikZ{
    \foreach \x in {1,0}
    {
    \foreach \y in {0}
    {
    \emptysquare{\x*0.5+\y*0.25}{\y*0.5};
    }
    }
    }}
    \right)\right).\label{eq:med_final}
\end{aligned}
\end{equation}

From Eq.~\eqref{eq:entropy_expression}, we can see that Eq.~\eqref{eq:consistent_global_state} satisfies the inequality in Eq.~\eqref{eq:med_final} with an equality. Therefore, we conclude 
\begin{equation}
\begin{aligned}
    \max_{P[\Omega]_{N\times M} \stackrel{c}{=} P[\miniblockfull]}H(P[\Omega]_{N\times M}) &= (N-1)(M-1)H\left(\Prob{\centeredTikZ{
    \foreach \x in {0, 1}
    {
    \foreach \y in {0, 1}
    {
    \emptysquare{\x*0.5+\y*0.25}{\y*0.5};
    }
    }
    }} \right) + 
    (N-2)(M-2)H\left( 
    \Prob{\centeredTikZ{
    \foreach \x in {0}
    {
    \foreach \y in {0}
    {
    \emptysquare{\x*0.5+\y*0.25}{\y*0.5};
    }
    }
    }}
    \right) \\
    &-\left(
    (N-2)(M-1)
    H\left( 
    \Prob{\centeredTikZ{
    \foreach \x in {0}
    {
    \foreach \y in {0,1}
    {
    \emptysquare{\x*0.5+\y*0.25}{\y*0.5};
    }
    }
    }}\right)
    + (N-1)(M-2)
    H\left( 
    \Prob{\centeredTikZ{
    \foreach \x in {1,0}
    {
    \foreach \y in {0}
    {
    \emptysquare{\x*0.5+\y*0.25}{\y*0.5};
    }
    }
    }}
    \right)\right).\label{eq:maxent_exact}
\end{aligned}
\end{equation}
This is precisely Theorem~\ref{thm:main2}, restated below.
\theoremtwo*

In particular, 
\begin{equation}
\boxed{
\max_{\substack{P[\Omega]_{N\times M} \stackrel{c}{=}P[\miniblockfull]}}
    \left(\lim_{N,M\to \infty} \frac{H(P[\Omega]_{N\times M})}{NM}\right)
    =H\left(\Prob{\centeredTikZ{
    \foreach \x in {0, 1}
    {
    \foreach \y in {0, 1}
    {
    \emptysquare{\x*0.5+\y*0.25}{\y*0.5};
    }
    }
    }} \right)-
    H\left(\Prob{\centeredTikZ{
    \foreach \x in {0, 1}
    {
    \foreach \y in {1}
    {
    \emptysquare{\x*0.5+\y*0.25}{\y*0.5};
    }
    }
    \emptysquare{1}{1};
    }
    } \right),
}\label{eq:max_entropy}
\end{equation}
where $P[\miniblockfull]$ is assumed to obey $\mathcal{C}_L$ and $\mathcal{C}_M$.

\subsection{Upper bounding thermodynamic free energy density}
Using the variational expression for the thermodynamic free energy, we can obtain the following upper bound.
\begin{equation}
\begin{aligned}
    F_{N,M}(\beta) &= \min_{P[\Omega]_{N\times M}} \left(\mathbb{E}_{P[\Omega]_{N\times M}}[E[\Omega]] - TH(P[\Omega]_{N\times M})  \right) \\ 
    &\leq \min_{P[\Omega]_{N\times M} \stackrel{c}{=}P[\miniblockfull]} \left(\mathbb{E}_{P[\Omega]_{N\times M}}[E[\Omega]] - TH(P[\Omega]_{N\times M})  \right).
\end{aligned}
\end{equation}

In the thermodynamic limit, by Eq.~\eqref{eq:max_entropy}, the free energy per $1\times 1$ cluster is upper bounded by
\begin{equation}
    \lim_{N,M\to \infty} \frac{F_{N,M}(\beta)}{NM} \leq  \mathbb{E}_{P[\miniblockfull]}[h_{\miniblockfull}] - \frac{1}{\beta}\left(H\left(\Prob{\centeredTikZ{
    \foreach \x in {0, 1}
    {
    \foreach \y in {0, 1}
    {
    \emptysquare{\x*0.5+\y*0.25}{\y*0.5};
    }
    }
    }} \right)-
    H\left(\Prob{\centeredTikZ{
    \foreach \x in {0, 1}
    {
    \foreach \y in {1}
    {
    \emptysquare{\x*0.5+\y*0.25}{\y*0.5};
    }
    }
    \emptysquare{1}{1};
    }
    } \right) \right),\label{eq:free_energy_upper}
\end{equation}
where the expectation value is taken over the probability distribution defined over the $2\times 2$ cluster, subject to the constraints $\mathcal{C}_L$ and $\mathcal{C}_M$. 

Thus, we conclude\footnote{It is interesting to note that a \emph{lower bound} to the thermodynamic free energy is formulated in terms of a similar optimization problem~\cite{Poulin2011}. An important difference is that our bound has an extra non-linear constraint $\mathcal{C}_M$ whereas Ref.~\cite{Poulin2011} only has a linear constraint, namely $\mathcal{C}_L$.}:
\begin{equation}
\boxed{
     \lim_{N,M\to \infty} \frac{F_{N,M}(\beta)}{NM} \leq 
    \min_{\substack{P[\miniblockfull]\text{ obeying}\\ \mathcal{C}_L,\text{ } \mathcal{C}_M }}
    \widetilde{f}_{\beta}\left(\Prob{\centeredTikZ{
    \foreach \x in {0, 1}
    {
    \foreach \y in {0, 1}
    {
    \emptysquare{\x*0.5+\y*0.25}{\y*0.5};
    }
    }
    }} \right),} \label{eq:sandwich}
\end{equation}
where 
\begin{equation}
    \widetilde{f}_{\beta} \left(\Prob{\centeredTikZ{
    \foreach \x in {0, 1}
    {
    \foreach \y in {0, 1}
    {
    \emptysquare{\x*0.5+\y*0.25}{\y*0.5};
    }
    }
    }} \right)=
     \mathbb{E}_{P[\miniblockfull]}[h_{\miniblockfull}] - \frac{1}{\beta}\left(H\left(\Prob{\centeredTikZ{
    \foreach \x in {0, 1}
    {
    \foreach \y in {0, 1}
    {
    \emptysquare{\x*0.5+\y*0.25}{\y*0.5};
    }
    }
    }} \right)-
    H\left(\Prob{\centeredTikZ{
    \foreach \x in {0, 1}
    {
    \foreach \y in {1}
    {
    \emptysquare{\x*0.5+\y*0.25}{\y*0.5};
    }
    }
    \emptysquare{1}{1};
    }
    } \right) \right).
\end{equation}

Eq.~\eqref{eq:sandwich} is the main result of this paper. We can \emph{upper bound} the thermodynamic free energy density by a solution to the optimization problem over an objective function $\widetilde{f}_{\beta}$. For a fixed translationally invariant Hamiltonian, one may hope to get progressively sharper bounds by increasing the size of the clusters.

\section{Discussion}
\label{sec:discussion}

We have proposed a new family of probability distributions that are ``exactly solvable.'' These distributions go beyond the mean-field approximation yet allow a closed-form expression for both the energy and the (maximum) entropy. By utilizing this ansatz, one may be able to tightly constrain the thermodynamic free energy density to a small interval, by combining the upper bound in Eq.~\eqref{eq:sandwich} with the well-known lower bounds based on convex programmings~\cite{Poulin2011}.

An important open question is the effect of error. In our work, we have provided an equality constraint on the set of marginals. What would happen if they are not satisfied exactly? There is a bound already proved in a quantum setting~\cite{Kim2016}, and we can expect a similar bound to apply here. The main idea is to  view Eq.~\eqref{eq:consistent_global_state} as a probability distribution created by a sequence of linear maps that act on a bounded-size region. The entropy of the global probability distribution can be bounded by invoking a continuity bound on entropies. Assuming that the error incurred in each step is $\epsilon,$ one can expect the total error accumulated on each marginal is $\mathcal{O}(NM\epsilon)$. So for finite $N$ and $M$, by first taking the $\epsilon \to 0$ limit, one may be able to show the robustness of our ansatz. However, the problem with this approach is that the total  effect of error diverges in the infinite volume limit, even for the upper bound on the thermodynamic free energy density. A bound that stays finite in the infinite volume limit will be clearly desirable.

\section*{Acknowledgement}
I thank Miguel Navascues and Mirjam Weilenmann for sharing their unpublished note, which contained a related result. I also thank them for encouraging me to submit this paper on arXiv. This work was supported by the Australian Research Council via the Centre of Excellence in Engineered Quantum Systems (EQUS) project number CE170100009.

\bibliographystyle{myhamsplain2}
\bibliography{bib}

\providecommand{\bysame}{\leavevmode\hbox to3em{\hrulefill}\thinspace}
\begin{thebibliography}{10}

\bibitem{Hastings1970}
W.~K. Hastings, \emph{Monte carlo sampling methods using markov chains and
  their applications}, Biometrika \textbf{57} (1970), no.~1, 97--109.

\bibitem{Kramers1941}
H.~A. Kramers and G.~H. Wannier, \emph{Statistics of the two-dimensional
  ferromagnet. part i}, Phys. Rev. \textbf{60} (1941), 252--262.

\bibitem{Kramers1941a}
H.~A. Kramers and G.~H. Wannier, \emph{Statistics of the two-dimensional
  ferromagnet. part ii}, Phys. Rev. \textbf{60} (1941), 263--276.

\bibitem{Bethe1935}
H.~A. Bethe, \emph{Statistical theory of superlattices}, Proceedings of the
  Royal Society of London. Series A - Mathematical and Physical Sciences
  \textbf{150} (1935), no.~871, 552--575.

\bibitem{Kikuchi1951}
R.~Kikuchi, \emph{A theory of cooperative phenomena}, Phys. Rev. \textbf{81}
  (1951), 988--1003.

\bibitem{Yedidia2005}
J.~S. {Yedidia}, W.~T. {Freeman}, and Y.~{Weiss}, \emph{Constructing
  free-energy approximations and generalized belief propagation algorithms},
  IEEE Transactions on Information Theory \textbf{51} (2005), no.~7,
  2282--2312.

\bibitem{Kim2016}
I.~H. Kim, \emph{Markovian marignals},  (2016),
  \href{http://arxiv.org/abs/1609.08579}{1609.08579}.

\bibitem{Kim2020}
I.~H. Kim, \emph{Entropy scaling law and the quantum marginal problem}, To
  appear. (2020).

\bibitem{Wang2018}
Z.~Wang and M.~Navascu\'es, \emph{Two-dimensional translation-invariant
  probability distributions: approximations, characterizations and no-go
  theorems}, Proc. R. Soc. Lond. \textbf{474} (2018), no.~2217, 20170822.

\bibitem{Fritz2013}
T.~{Fritz} and R.~{Chaves}, \emph{Entropic inequalities and marginal problems},
  IEEE Transactions on Information Theory \textbf{59} (2013), no.~2, 803--817.

\bibitem{Poulin2011}
D.~Poulin and M.~B. Hastings, \emph{Markov entropy decomposition: A variational
  dual for quantum belief propagation}, Phys. Rev. Lett. \textbf{106} (2011),
  080403.

\end{thebibliography}
\end{document}